\documentclass{article}

\usepackage{graphicx}
\usepackage{latexsym}
\usepackage{amsmath, amsfonts, amssymb, amsthm}
\newtheorem{theorem}{Theorem}
\newtheorem{definition}{Definition}

\newtheorem{proposition}{Proposition}

\begin{document}

\title{Relationship between mis\`{e}re {\sc nim} and two-player {\sc goishi hiroi}}

\author{Tomoaki Abuku\footnote{Gifu University,  buku3416@gmail.com}, Masanori Fukui\footnote{Faculty of Software and Information Science,
Iwate Prefectural University, fukui\_m@iwate-pu.ac.jp}, Shin-ichi Katayama\footnote{Institute of Liberal Arts and Sciences, Tokushima University, shinkatayama@tokushima-u.ac.jp}, Koki Suetsugu\footnote{Waseda University and Hiroshima University, suetsugu.koki@gmail.com}}

\maketitle

\begin{abstract}
In combinatorial game theory, there are two famous winning conventions, normal play and mis\`{e}re play.
Under normal play convention, the winner is the player who moves last and under mis\`{e}re play convention, the loser is the player who moves last. The difference makes these conventions completely different, and usually, games under mis\`{e}re play convention is much difficult to analyze than games under normal play convention.
In this study, we show an interesting relationship between rulesets under different winning conventions; we can determine the winner of two-player {\sc goishi hiroi} under normal play convention by using {\sc nim} under mis\`{e}re play convention. We also analyze two-player {\sc goishi hiroi} under mis\`{e}re play convention.
\end{abstract}

\section{Introduction}
In this paper, we focus on two famous winning conventions of combinatorial games, {\em normal play} convention and {\em mis\`{e}re play} convention. The former means  ``the player who moves last is the winner'' and the later means ``the player who moves last is the loser.''
Usually, the two conventions are completely different and even if a method to determine the player who has a winning strategy is found under one convention, one cannot determine the winner for another convention. Thus, the relationship between two conventions are investigated with interest.


In this paper, we show that there is the same construction in two rulesets under different conventions.

\subsection{Impartial games}
Combinatorial game theory studies two-player perfect information games with no chance. We also assume that in this study, the rulesets are {\em short}, that is the set of options is finite in any position, and the play must end in finite moves.

A ruleset is called {\em impartial} if in every position, the sets of options of both players are the same. One of the most famous impartial ruleset is {\sc nim} \cite{Bou02}.
In {\sc nim}, there are some piles of stones and in each turn, the current player chooses one pile and removes any number of stones.

For impartial games under normal play convention, we can define {\em Sprague-Grundy value} for every position.

Let $\mathbb{N}_0$ be the set of all nonnegative integers.
\begin{definition}
    Let $S$ be a finite set of integers.
    The minimum excluded integer of $S$ is as follows: 
    $$
    {\rm mex}(S) = \min (\mathbb{N}_0 \setminus S).
    $$
\end{definition}
Note that ${\rm mex}(S)$ is defined even if $S$ includes some negative integers. For example, if $S$ includes not only nonnegative integers but also $-1,$ we let $S' = S \setminus \{-1\}$ and ${\rm mex}(S) = \min(\mathbb{N}_0 \setminus S')$. Usually, we only consider the case that $S$ is a finite set of nonnegative integers, but in this paper, we consider the case that $S$ has $-1$ as an element. 

\begin{definition}
    For any position $G,$ its Sprague-Grundy value $\mathcal{G}(G)$ is as follows:
    $$
    \mathcal{G}(G) = {\rm mex}(\{\mathcal{G}(G') \mid G' \text{ is an option of } G\}).
    $$
\end{definition}

When a position is separated into independent components, that is, a player, on their turn, chooses exactly one component and makes one move, we say that the position is a {\em disjunctive sum} of the components. The disjunctive sum of the positions $G$ and $H$ is denoted as $G+H$.

We call a position a $\mathcal{P}$-position if the previous player has a winning strategy. Also, we call a position an $\mathcal{N}$-position if the next player has a winning strategy.

If the set of all positions can be separated into two sets $N$ and $P$ such that 
\begin{itemize}
    \item every position in $N$ has an option in $P$, and
    \item every position in $P$ has no option in $P$,
\end{itemize}

then under normal play convention, $N$ and $P$ are the sets of $\mathcal{N}$-positions and $\mathcal{P}$-positions, respectively. Note that from the second condition, every terminal position is in $P$.

Also, if the set of all positions can be separated into two sets $N'$ and $P'$ such that 
\begin{itemize}
    \item every terminal position is in $N'$,
    \item every non-terminal position in $N'$ has an option in $P'$, and
    \item every position in $P'$ has no option in $P'$,
\end{itemize}
then under mis\`{e}re play convention, $N'$ and $P'$ are the sets of $\mathcal{N}$-positions and $\mathcal{P}$-positions, respectively.

Sprague and Grundy independently showed that Sprague-Grundy value can be used to determine which player has a winning strategy. In addition, the Sprague-Grundy value of a disjunctive sum of components can be calculated by using each component's Sprague-Grundy value and XOR operator.

\begin{theorem}[Sprague \cite{Spr35}, Grundy \cite{Gru39}]
\label{Thmsg}
Let $G$ be a position of an impartial ruleset under normal play convention. The following (1), (2), and (3) hold.
\begin{enumerate}
    \item $\mathcal{G}(G) = 0$ if and only if $G$ is a $\mathcal{P}$-position.
    \item $\mathcal{G}(G) \neq 0$ if and only if $G$ is an $\mathcal{N}$-position.
    \item $\mathcal{G}(G+H) = \mathcal{G}(G) \oplus \mathcal{G}(H)$.
\end{enumerate}
Here, $\oplus$ is the XOR operator for binary notation.

\end{theorem}
From this theorem, we can determine which player has a winning strategy from Sprague-Grundy values of the components of disjunctive sum if it is under normal play convention.

Let $(x, y)$ be a pair of nonnegative integers. We also let $G_0$ as 
\begin{eqnarray*}
    G_0((x,y)) = \left \{ \begin{array}{cc}
        0 & (x+y = 0)  \\
        \begin{array}{c}{\rm mex}(\{G_0((x',y)) \mid 0\leq x'<x\}   \\ \cup \{G_0((x,y'))\mid 0\leq y'<y\}) \end{array} & (x + y > 0).
    \end{array} \right.
\end{eqnarray*}
 We will use single brackets like  $G(x, y)$ instead of double brackets $G((x, y))$ if there is no confusion.
Table \ref{tab:G0} shows the values of $G_0(x,y)$.
\begin{table}[htb]
    \centering
    \begin{tabular}{c|cccccccccccc}
         & 0 & 1 & 2 & 3 & 4 & 5 & 6 & 7 & 8 & 9 & 10 & 11 \\ \hline
        0 & 0 & 1 & 2 & 3 & 4 & 5 & 6  & 7 & 8 &9  &10 &11  \\
        1 & 1 & 0 & 3 & 2 & 5 & 4 &7 & 6 & 9 & 8 & 11 & 10 \\
        2 &2 & 3 & 0 & 1 & 6 & 7 &4 & 5 & 10 & 11 & 8 & 9 \\
        3 &3 & 2 & 1 & 0 & 7 & 6 &5 & 4 & 11 & 10 & 9 & 8\\
        4 &4 & 5 & 6 & 7 & 0 & 1 &2 & 3 & 12 & 13 & 14 & 15 \\
        5 &5 & 4 & 7 & 6 & 1 & 0 &3 & 2 & 13 & 12 & 15 & 14 \\
        6 &6 & 7 & 4 & 5 & 2 & 3 &0 & 1 & 14 & 15 & 12 & 13 \\
        7 &7 & 6 & 5 & 4 & 3 & 2 &1 & 0 & 15 & 14 & 13 & 12 \\
        8 & 8 & 9 & 10 & 11 & 12 & 13 & 14 & 15 & 0 & 1 & 2 & 3 \\
        9 & 9 & 8 & 11 & 10 & 13 & 12 & 15 & 14 & 1 & 0 & 3 & 2 \\
        10 & 10 & 11 & 8 & 9 & 14 &15 & 12 & 13 & 2 & 3 & 0 & 1 \\
        11 & 11 & 10 & 9 & 8 & 15 & 14 & 13 & 12 & 3 & 2 & 1 & 0 
    \end{tabular}
    \caption{Values of $G_0(x,y)$}
    \label{tab:G0}
    \end{table}
   
The function $G_0$ is the Sprague-Grundy value of a position $(x, y)$ in two-pile {\sc nim}. Since two-pile {\sc nim} can be considered as a disjunctive sum of single piles, $G_0(x,y) = x\oplus y$ holds. In three-pile {\sc nim}, $(x, y, z)$ is a $\mathcal{P}$-position if and only if $z = G_0(x,y)$.

Next, we consider mis\`{e}re games.

Under normal play convention, the terminal position is a $\mathcal{P}$-position, and its Sprague-Grundy value is $0$. On the other hand, under mis\`{e}re play convention, the terminal position is an $\mathcal{N}$-position. Therefore, to obtain the same corresponding between values and outcomes as Theorem \ref{Thmsg}, the terminal position must have non-zero Sprague-Grundy value.

One way is putting $1$ for terminal positions and for other positions, using the mex-rule as the case of normal play as following $G_1$:
    \begin{eqnarray*}
    G_{1}((x, y)) = \left\{ \begin{array}{cc}
        1 & (x+  y= 0) \\
         \begin{array}{c}{\rm mex}(\{G_{1}((x',y)) \mid 0\leq x'<x\} \\ \cup \{G_{1}((x,y'))\mid 0\leq y'<y\})\end{array} & (x + y> 0) .
    \end{array}\right.
    \end{eqnarray*}
Table \ref{tab:G1} shows the values of $G_1(x,y)$.
\begin{table}[htb]
    \centering
    \begin{tabular}{c|cccccccccccc}
         & 0 & 1 & 2 & 3 & 4 & 5 & 6 & 7 & 8 & 9 & 10 & 11 \\ \hline
        0 & 1 & 0 & 2 & 3 & 4 & 5 & 6  & 7 & 8 &9  &10 &11  \\
        1 & 0 & 1 & 3 & 2 & 5 & 4 &7 & 6 & 9 & 8 & 11 & 10 \\
        2 &2 & 3 & 0 & 1 & 6 & 7 &4 & 5 & 10 & 11 & 8 & 9 \\
        3 &3 & 2 & 1 & 0 & 7 & 6 &5 & 4 & 11 & 10 & 9 & 8\\
        4 &4 & 5 & 6 & 7 & 0 & 1 &2 & 3 & 12 & 13 & 14 & 15 \\
        5 &5 & 4 & 7 & 6 & 1 & 0 &3 & 2 & 13 & 12 & 15 & 14 \\
        6 &6 & 7 & 4 & 5 & 2 & 3 &0 & 1 & 14 & 15 & 12 & 13 \\
        7 &7 & 6 & 5 & 4 & 3 & 2 &1 & 0 & 15 & 14 & 13 & 12 \\
        8 & 8 & 9 & 10 & 11 & 12 & 13 & 14 & 15 & 0 & 1 & 2 & 3 \\
        9 & 9 & 8 & 11 & 10 & 13 & 12 & 15 & 14 & 1 & 0 & 3 & 2 \\
        10 & 10 & 11 & 8 & 9 & 14 &15 & 12 & 13 & 2 & 3 & 0 & 1 \\
        11 & 11 & 10 & 9 & 8 & 15 & 14 & 13 & 12 & 3 & 2 & 1 & 0 
    \end{tabular}
    \caption{Values of $G_1(x,y)$}
    \label{tab:G1}
    \end{table}

It is easy to find that $G_0(0, 0) = G_0(1,1) = G_1(0, 1) = G_1(1, 0) = 0, G_0(0, 1)= G_0(1, 0) = G_1(0, 0) = G_1(1, 1) = 1$ and for any other $(x, y),$ $G_0(x, y) = G_1(x, y)$.

Such rulesets, that is only some positions' $0$ and $1$ are swapped, are known to easy to analyze (for the details, see chapter 5 of \cite{Sie}).
Therefore, setting $1$ for terminal positions is meaningful. However, in this study, we consider another way. That is, putting $-1$ for the terminal position of mis\`{e}re {\sc nim} as following $G_{-1}$:

    \begin{eqnarray*}
    G_{-1}((x, y)) = \left\{ \begin{array}{cc}
        -1 & (x+  y= 0) \\
         \begin{array}{c}{\rm mex}(\{G_{-1}((x',y)) \mid 0\leq x'<x\} \\ \cup \{G_{-1}((x,y'))\mid 0\leq y'<y\})\end{array} & (x + y> 0) ,
    \end{array}\right.
    \end{eqnarray*}

Two-heap {\sc nim} under mis\`{e}re play can be considered as a situation such that any move to $(0, 0)$ is forbidden in two-heap {\sc nim} under normal play. In this sense, $G_{-1}$ is the function of Sprague-Grundy values of this ruleset.  
We also define a function $G^*_{-1}$ as 

\begin{eqnarray*}
    G^*_{-1}((x,y)) = \left \{
    \begin{array}{cc}
        0 & (x+ y=0) \\
        -1 &(x+ y = 1) \\
        \begin{array}{c}{\rm mex}(\{G^*_{-1}((x',y)) \mid 0\leq x'<x\} \\ \cup \{G^*_{-1}((x,y'))\mid 0\leq y'<y\})\end{array} & (x + y > 1).
    \end{array}\right.
    \end{eqnarray*}

  Tables \ref{tab:Ginf}, and \ref{tab:Ginfinf} show the values of $ G_{-1}(x,y),$ and $G^*_{-1}(x,y),$ respectively.
 $G^*_{-1}$ is Sprague-Grundy values for two-heap {\sc nim} under normal play convention where moves to $(0, 1)$ and moves to $(1, 0)$ are forbidden.
  $G^*_{-1}$ has the same $0$-positions as $G_0$, but this has two $-1$s from its definition; thus many values are different from $G_0$.

    \begin{table}
    \centering
    \begin{tabular}{c|cccccccccccc}
         & 0 & 1 & 2 & 3 & 4 & 5 & 6 & 7 & 8 & 9 & 10 & 11 \\ \hline
        0 & $-1$ & 0 & 1 & 2 & 3 & 4 & 5  & 6 & 7 & 8 & 9 & 10 \\
        1 & 0 & 1 & 2 & 3 & 4 & 5 &6 & 7 & 8 & 9 & 10 & 11 \\
        2 &1 & 2 & 0 & 4 & 5 & 3 &7 & 8 & 6 & 10 & 11 & 9 \\
        3 &2 & 3 & 4 & 0 & 1 & 6 &8 & 5 & 9 & 7 & 12 & 13 \\
        4 &3 & 4 & 5 & 1 & 0 & 2 &9 & 10 & 11 & 6  & 7 & 8\\
        5 &4 & 5 & 3 & 6 & 2 & 0 &1 & 9 & 10 & 11 & 8 & 7\\
        6 &5 & 6 & 7 & 8 & 9 & 1 &0 & 2 & 3 & 4 & 13 & 12 \\
        7 &6 & 7 & 8 & 5 & 10 & 9 &2 & 0 & 1 & 3 & 4 & 14 \\
        8 & 7 & 8 & 6 & 9 & 11 & 10 & 3 & 1 & 0 & 2 & 5 & 4 \\
        9 & 8 & 9 & 10 & 7 & 6 & 11 & 4 & 3 & 2 & 0 & 1 & 5 \\
        10 & 9 & 10 & 11 & 12 & 7 & 8 & 13 & 4 & 5 & 1 & 0 & 2\\
        11 & 10 & 11 & 9 & 13 & 8 & 7 & 12 & 14 & 4 & 5 & 2 & 0
    \end{tabular}
    \caption{Values of $G_{-1}(x,y)$}
    \label{tab:Ginf}
    \end{table}
    \begin{table}
    \centering
    \begin{tabular}{c|cccccccccccc}
         & 0 & 1 & 2 & 3 & 4 & 5 & 6 & 7 & 8 & 9 & 10 & 11 \\ \hline
        0 & 0 & $-1$ & 1 & 2 & 3 & 4 & 5  & 6 & 7 & 8 & 9 & 10 \\
        1 & $-1$ & 0 & 2 & 1 & 4 & 3 &6 & 5 & 8 & 7 & 10 & 9 \\
        2 &1 & 2 & 0 & 3 & 5 & 6 &4 & 7 & 9 & 10 & 8 & 11 \\
        3 &2 & 1 & 3 & 0 & 6 & 5 &7 & 4 & 10 & 9 & 11 & 8\\ 
        4 &3 & 4 & 5 & 6 & 0 & 1 &2 & 8 & 11 & 12 & 7 & 13 \\
        5 &4 & 3 & 6 & 5 & 1 & 0 &8 & 2 & 12 & 11 & 13 & 7 \\
        6 &5 & 6 & 4 & 7 & 2 & 8 &0 & 1 & 3 & 13 & 12 & 14 \\
        7 &6 & 5 & 7 & 4 & 8 & 2 &1 & 0 & 13 & 3 & 14 & 12  \\
        8 & 7 & 8 & 9 & 10 & 11 & 12 & 3 & 13 & 0 & 1 & 2 & 4 \\
        9 & 8 & 7 & 10 & 9 & 12 & 11 & 13 & 3 & 1 & 0 & 4 & 2 \\
        10 & 9 & 10 & 8 & 11 & 7 & 13 & 12 & 14 & 2 & 4 & 0 & 1 \\
        11 & 10 & 9 & 11 & 8 & 13 & 7 & 14 & 12 & 4 & 2 & 1 & 0
    \end{tabular}
    \caption{Values of $G^*_{-1}(x,y)$}
    \label{tab:Ginfinf}
\end{table}

In this paper, we show that these functions can be used to analyze another ruleset, two-player {\sc goishi hiroi}.

\subsection{{\sc Goishi hiroi} and two-player {\sc goishi hiroi}}

We consider two-player {\sc goishi hiroi}. {\sc Goishi hiroi} (or, {\sc hiroimono}) is a traditional Japanese puzzle (\cite{Tag}).
In this puzzle, some stones are arranged on a lattice board. 
The player can do the following move until Step $n$ for arbitrary $n \geq 0$.
\begin{itemize}
    \item [Step $0$] : The player picks up one stone. Then, they choose one direction from up, down, left, and right.
    \item [Step $n$] : The player traces the lattice board from the stone picked up in Step $n-1$ in the current direction until reaches another stone and picks up the stone. Then, they choose one direction except for the reverse direction of the current direction.
\end{itemize}
If a player can pick up all stones in one move, we say that the player has solved the given puzzle. 
Figure \ref{fig:puzzleexp} shows one problem and solution in {\sc goishi hiroi}.
The complexity of {\sc goishi hiroi} has been studied in \cite{And07} and \cite{FSS17}.

It is known that two-player {\sc goishi hiroi} was also played, but the details of the rule had been lost.
From the remaining literature, we propose the following rules:
\begin{itemize}
    \item Some stones colored black or white are on a lattice board. 
    \item Each player, on their turn, do the following move until Step $n$ for arbitrary $ n \geq 0$.
    \begin{itemize}
    \item [Step $0$]: The current player picks up one stone. Then, they choose one direction from up, down, left, and right.
    \item [Step $n$]: The current player traces the lattice board from the stone picked up in Step $n-1$ in the current direction until reaches another stone. The player picks up the stone if the stone has the same color as the previous stone, and otherwise, the move ends. If the player could pick up a stone, they choose one direction except for the reverse direction of the current direction.
\end{itemize}
\item When all stones have been removed, the game ends.
\end{itemize}
The general case of this ruleset is quite complex, so in this study, we consider a restricted version of this ruleset.
We assume that all the stones are on the same line. In addition, at most two pairs of adjacent stones have different colors, so each position can be denoted as $(x, y, z)(x, y, z \in \mathbb{N}_0).$

Figure \ref{fig:positionexmp} shows the position $(x, y, z)$ in this ruleset and Figure \ref{fig:sequence} shows a sequence of moves on this ruleset. This ruleset is similar to three-pile {\sc nim.} However, in this ruleset, when the middle pile becomes empty, the current player can immediately win by removing all stones from both remaining piles. 

Note that this linear version is considered in another study \cite{MYS24} as ``StrNim''. However, in the study, the relationship between this ruleset and mis\`{e}re {\sc nim}, which is discussed here, is not considered.
\begin{figure}[htb]
    \centering
    \includegraphics[width = 6cm]{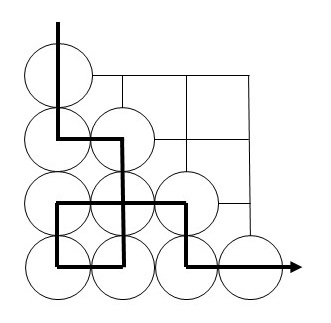}
    \caption{Puzzle and solution in {\sc goishi hiroi}}
    \label{fig:puzzleexp}
\end{figure}
\begin{figure}
    \centering
\includegraphics[width = 9cm]{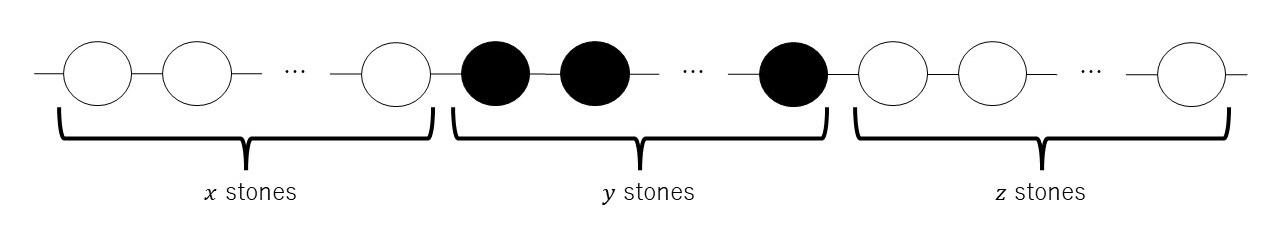}
    \caption{Position $(x, y, z)$}
    \label{fig:positionexmp}
\end{figure}

\begin{figure}[htb]
    \centering
\includegraphics[width = 9cm]{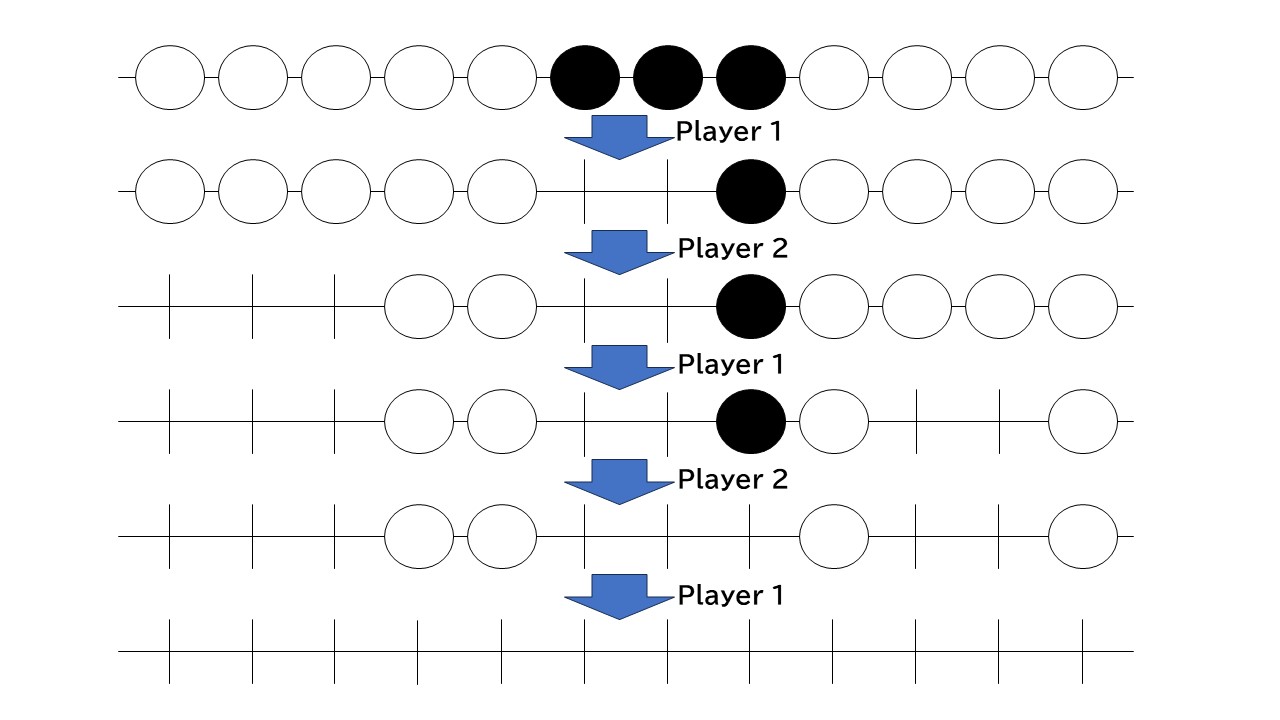}
    \caption{Sequence of moves: $(5, 3, 4) \rightarrow (5, 1, 4) \rightarrow (2, 1, 4) \rightarrow (2, 1, 2) \rightarrow (2, 0 , 2) \rightarrow (0, 0, 0)$} 
    \label{fig:sequence}
\end{figure}

\section{Main results}

We show a relationship between two-player {\sc goishi hiroi} under normal play convention and mis\`{e}re {\sc nim} as follows:

\begin{theorem}
    A position $(x, y, z)$ in two-player {\sc goishi hiroi} under normal play convention is a $\mathcal{P}$-position if and only if $y = G_{-1}(x, z) + 1$. 
\end{theorem}
\begin{proof}

We prove this by induction on $x + y + z$.
Let $N = \{(x, y, z) \mid y \neq G_{-1}(x, z) + 1\}$ and $P = \{(x, y, z) \mid y = G_{-1}(x, z)+1\}$. We prove that every position $g \in N$ has an option $g' \in P$ and every position $g \in P$ has no option $g' \in P$. It is enough to say that $N$ (resp. $P$) is the set of $\mathcal{N}$-positions (resp. $\mathcal{P}$-positions).

If $x + y + z = 0$, then $(x, y, z) \in P$ and has no option in $N$.

Assume that $x + y + z > 0$ and $g = (x, y, z) \in P$. Then, from the definition of $G_{-1}, y > 0$.
We consider an option $g'$ of $g$.
If $g' = (x', y, z)(x' < x)$, then $G_{-1}(x', z) \neq G_{-1}(x, z)$
 since $G_{-1}(x, z) = {\rm mex}(\{G_{-1}(x', z) \mid x' < x\} \cup \{G_{-1}(x, z') \mid z' < z\}).$ Therefore, $y \neq G_{-1}(x', z) + 1$.
 Similarly, if $g' = (x, y, z')(z' < z),$ then $y \neq G_{-1}(x, z') + 1.$
Furthermore, if $g' = (x, y', z)(y' < y),$ it is trivial $y' \neq G_{-1}(x, z) + 1$ because $y' < y$.

Therefore, every position $g \in P$ has no option in $P$.

Then, we consider the case that $x + y + z > 0$ and $g = (x, y, z) \in N$.
 If $y > G_{-1}(x, z) + 1, g$ has an option $g' = (x, G_{-1}(x, z)+ 1, z) \in P.$ Thus, we assume that $y < G_{-1}(x, z) + 1.$ If $y = 0, g$ has an option $g' = (0, 0, 0) \in P.$ Otherwise, since $G_{-1}(x, z) = {\rm mex} (\{G_{-1}(x',z) \mid 0\leq x'<x\}  \cup \{G_{-1}(x,z')\mid 0\leq z'<z\})$ we can assume that without loss of generality, there exists $x' < x$ such that $G_{-1}(x',z) = y- 1$. Thus, $g$ has an option $g' = (x', y, z)$ such that $y = G_{-1}(x', z) + 1$.

Therefore, $P$ is the set of all $\mathcal{P}$-positions and $N$ is the set of all $\mathcal{N}$-positions.
\end{proof}

We also show a relationship between two-player {\sc goishi hiroi} under mis\`{e}re play convention and function $G^*_{-1}$.

\begin{theorem}
    A position $(x, y, z)$ in two-player {\sc goishi hiroi} under mis\`{e}re play convention is a $\mathcal{P}$-position if and only if $y = G^*_{-1}(x, z) + 1$.
\end{theorem}
\begin{proof}
    We prove this by induction on $x + y + z$. Let $N = \{(x, y, z) \mid y \neq G^*_{-1}(x,z) + 1\}$ and $P = \{(x, y, z) \mid y = G^*_{-1}(x,z)+1\}$.
    We prove that every position $g \in N$  has an option $g' \in P$ or $g = (0, 0, 0)$ and every position $g \in P$ has no option $g' \in P$.
    It is enough to say that $N$ (resp. $P$) is the set of $\mathcal{N}$-positions (resp. $\mathcal{P}$-positions).

    If $(x ,y ,z) = (0, 0, 0),$ then $y \neq G^*_{-1}(x,z) + 1,$ and $(x, y ,z) \in N$.
    If $(x, y, z) = (1, 0, 0), (x, y, z) = (0, 1, 0),$ or $(x, y ,z)=(0, 0, 1),$ then $y = G^*_{-1}(x,z) + 1,$ and $(x, y ,z) \in P$. In these cases, every position has only one option $(0, 0, 0)$ which belongs to $N$.

Assume that $x + y + z > 1$ and $g = (x, y, z) \in P$. Then, from the definition of $G^*_{-1}, y > 0$.
We consider an option $g'$ of $g$.
If $g' = (x', y, z)(x' < x)$, then $G^*_{-1}(x', z) \neq G^*_{-1}(x, z)$
 since $G^*_{-1}(x, z) = {\rm mex}(\{G^*_{-1}(x', z) \mid x' < x\} \cup \{G^*_{-1}(x, z') \mid z' < z\}).$ Therefore, $y \neq G^*_{-1}(x', z) + 1$.
 Similarly, if $g' = (x, y, z')(z' < z),$ then $y \neq G^*_{-1}(x, z') + 1.$
Furthermore, if $g' = (x, y', z)(y' < y),$ it is trivial $y' \neq G^*_{-1}(x, z) + 1$ because $y' < y$.

Therefore, every position $g \in P$ has no option in $P$.

Then, we consider the case that $x + y + z > 1$ and $g = (x, y, z) \in N$.
If $y > G^*_{-1}(x, z) + 1, g$ has an option $g' = (x, G^*_{-1}(x, z)+ 1, z) \in P.$ Thus we assume that $y < G^*_{-1}(x, z) + 1.$ If $y = 0, g$ has an option $g' = (1, 0, 0) \in P$ or $g'' = (0,0,1) \in P$. Otherwise, since $G^*_{-1}(x, z) = {\rm mex} (\{G^*_{-1}(x',z) \mid 0\leq x'<x\}  \cup \{G^*_{-1}(x,z')\mid 0\leq z'<z\}),$ we can assume that without loss of generality, there exists $x' < x$ such that $G^*_{-1}(x',z) = y- 1$. Thus, $g$ has an option $g' = (x', y, z)$ such that $y = G^*_{-1}(x', z) + 1$.

Therefore, $P$ is the set of all $\mathcal{P}$-positions and $N$ is the set of all $\mathcal{N}$-positions.
    
\end{proof}

Unfortunately, unlike the function $G_0$, it seems difficult to find closed formulas for the functions $G_{-1}$ and $G^*_{-1}$.
However, we find necessary and sufficient conditions for the cases that the return values of the functions are small.
\begin{theorem}
The following (1), (2), (3), and (4) hold.
\begin{enumerate}
    \item 
    Let $A_0 = \{(n,n)\mid n \geq 2\}$. Then
    $G_{-1}(x,y) = 0$ if and only if $(x, y) \in A_0 \cup  \{(0, 1), (1, 0)\}$.
    \item Let $A_1 = \{(2n, 2n-1), (2n - 1, 2n) \mid n \geq 2\}$. 
    Then $G_{-1}(x,y) = 1$ if and only if $(x, y) \in A_1 \cup \{(0, 2), (1, 1), (2, 0)\}.$
    \item Let $A_2 = \{(2n, 2n+1), (2n+1, 2n) \mid n \geq 2\}$.
    Then $G_{-1}(x,y) = 2$ if and only if $(x, y) \in A_2 \cup \{(0,3), (1,2), (2,1), (3,0)\}.$
    \item Let $A_3 = \{(4n-2, 4n), (4n-1, 4n+1), (4n, 4n-2), (4n +1, 4n-1)  \mid n\geq 2\}.$
    Then $G_{-1}(x,y) = 3$ if and only if $(x, y) \in A_3 \cup \{(0,4), (1,3), (2, 5), (3,1), (4,0), (5, 2)\}.$
    
\end{enumerate}
\end{theorem}
\begin{proof}
We prove by induction.
\begin{enumerate}
    \item It is easy to confirm that $G_{-1}(0,1) = G_{-1}(1, 0) = 0$.
    Let $n \geq 2$. For any $n'<n$, $G_{-1}(n, n') = {\rm mex}(\{G_{-1}(x', n') \mid 0 \leq x' < n\} \cup \{G_{-1}(n, y') \mid 0 \leq y' < n'\})\neq 0$ because $G_{-1}(n', n') = 0$ if $n' \geq 2$ and $G_{-1}(0, 1) = G_{-1}(1,0)= 0.$ In a similar way, for any $n' < n, G_{-1}(n', n) \neq 0$. Therefore, $G_{-1}(n, n) = {\rm mex}(\{G_{-1}(n', n) \mid 0\leq n' < n \} \cup \{G_{-1}(n, n') \mid 0\leq n' < n\}) = 0.$
    \item It is easy to confirm that $G_{-1}(0, 2) = G_{-1}(1, 1) =  G_{-1}(2, 0) = 1$.
    Let $n \geq 2.$
    For any $n' < n, G_{-1}(2n, 2n') \neq 1$ and $G_{-1}(2n, 2n' - 1) \neq 1$ because $G_{-1}(2n'-1, 2n') = 1$ if $n' \geq 2, G_{-1}(2n', 2n' - 1) = 1$ if $n' \geq 2$, and $G_{-1}(0, 2) = G_{-1}(1, 1) = G_{-1}(2, 0) = 1.$ Similarly, for any $n' < n, G_{-1}(2n', 2n-1) \neq 1$ and $G_{-1}(2n' - 1, 2n-1) \neq 1$ hold. In addition, from (1), $G_{-1}(2n - 1, 2n - 1) = 0$. Therefore, $G_{-1}(2n, 2n-1) = {\rm mex}(\{G_{-1}(x', 2n-1)\mid 0\leq x' < 2n\} \cup \{G_{-1}(2n, y') \mid 0 \leq y' < 2n - 1)\}) = 1$. We can also show $G_{-1}(2n-1, 2n) = 1$ by symmetry.    
    \item It is easy to confirm that $G_{-1}(0, 3) = G_{-1}(1, 2) = G_{-1}(2, 1) = G_{-1}(3, 0) =  2$. 
    Let $n \geq 2$.
    For any $n' < n, G_{-1}(2n , 2n') \neq 2$ and $G_{-1}(2n, 2n' + 1) \neq 2$ because $G_{-1}(2n'+1, 2n') = 2$ if $n' \geq 2,$  $G_{-1}(2n', 2n' + 1) = 2$ if $n' \geq 2,$ and $G_{-1}(0, 3) = G_{-1}(1,2 ) = G_{-1}(2, 1) = G_{-1}(3, 0) = 2$. Similarly, for any $n' < n, G_{-1}(2n', 2n+ 1) \neq 2$ and $G_{-1}(2n'+1, 2n+1) \neq 2$ hold. In addition, from (1), $G_{-1}(2n, 2n) = 0$ and from (2), $G_{-1}(2n, 2n-1) = 1$. Therefore, $G_{-1}(2n, 2n+1) = {\rm mex}(\{G_{-1}(x', 2n+1) \mid 0 \leq x' < 2n\} \cup \{G_{-1}(2n, y') \mid 0 \leq y' < 2n + 1\}) = 2$. We can also show $G_{-1}(2n + 1, 2n) = 2$ by symmetry.
    
    \item It is easy to confirm that $G_{-1}(0, 4) = G_{-1}(1, 3) = G_{-1}(2, 5) = G_{-1}(3, 1) = G_{-1}(4, 0) = G_{-1}(5, 2) = 3.$ 
    Let $n \geq 2$.
    For any $n' < n, G_{-1}(4n - 2, 4n' - 2) \neq 3, G_{-1}(4n - 2, 4n' - 1) \neq 3, G_{-1}(4n - 2, 4n') \neq 3,$ and $G_{-1}(4n - 2, 4n' + 1) \neq 3$ because $G_{-1}(4n', 4n' - 2) = 3$ if $n' \geq 2, G_{-1}(4n' +1, 4n' - 1) = 3$ if $n' \geq 2, G_{-1} (4n' - 2, 4n') = 3$ if $n' \geq 2, G_{-1}(4n' - 1, 4n' + 1) = 3$ if $n' \geq 2, $ and  $G_{-1}(2, 5)= G_{-1}(0, 4) = G_{-1}(1, 3) = G_{-1}(5, 2) = G_{-1}(3, 1) = G_{-1}(4, 0) = 3.$ Similarly, for any $n'< n, G_{-1}(4n' - 2, 4n) \neq 3, G_{-1}(4n' - 1, 4n) \neq 3, G_{-1}(4n', 4n) \neq 3,$ and $G_{-1}(4n' + 1, 4n) \neq 3$ hold.
    In addition, from (1), $G_{-1}(4n-2, 4n-2) = 0,$ from (2), $G_{-1}(4n - 2, 4n - 3) = 1,$ and from (3), $G_{-1}(4n-2, 4n-1) = 2.$
    Therefore, $G_{-1}(4n-2, 4n) = {\rm mex}(\{G_{-1}(x', 4n) \mid 0 \leq x' < 4n - 2\} \cup \{G_{-1}(4n-2, y') \mid 0 \leq y' < 4n\}) = 3.$
    We can also show $G_{-1}(4n-1,4n+1) = G_{-1}(4n, 4n- 2) = G_{-1}(4n + 1, 4n - 1) = 3$ in similar ways.

\end{enumerate}    
\end{proof}

\begin{theorem}
The following (1), (2), (3), and (4) hold.
\begin{enumerate}
    \item 
    Let $B_0 = \{(n,n)\mid n \geq 0\}$. Then $G^*_{-1}(x,y) = 0$ if and only if $(x, y) \in B_0$.
    \item Let $B_1 = \{(2n, 2n+1), (2n+1, 2n) \mid n \geq 2\}$.
    Then $G^*_{-1}(x,y) = 1$ if and only if $(x, y) \in B_1 \cup \{(0,2), (1,3), (2,0), (3,1)\}$
    \item Let $B_2 = \{(4n+2, 4n), (4n+3, 4n+1), (4n, 4n+2), (4n+1, 4n+3) \mid n \geq 1\}.$
    Then $G^*_{-1}(x,y) = 2$ if and only if $(x, y) \in B_2 \cup \{(0,3), (1,2), (2, 1), (3, 0)\}$.
    \item Let $B_3 = \{(4n-2, 4n), (4n-1, 4n+1), (4n, 4n-2), (4n+1, 4n-1)\mid n\geq 2\}.$
    Then $G^*_{-1}(x,y) = 3$ if and only if $(x, y) \in B_3 \cup \{(0,4),(1,5), (2,3),(3,2),(4,0),(5,1)\}.$
\end{enumerate}
\end{theorem}
\begin{proof}
We prove by induction.
\begin{enumerate}
    \item It is easy to confirm that $G^*_{-1}(0,0) = 0$.
    Let $n \geq 1$.
    For any $n'<n$, $G^*_{-1}(n, n') = {\rm mex}(\{G^*_{-1}(x', n') \mid 0 \leq x' < n\} \cup \{G^*_{-1}(n, y') \mid 0 \leq y' < n'\})\neq 0$ because  $G_{-1}(n', n') = 0.$  In a similar way, for any $n' < n, G^*_{-1}(n', n) \neq 0$. Therefore, $G^*_{-1}(n, n) = {\rm mex}(\{G^*_{-1}(n', n) \mid 0\leq n' < n \} \cup \{G^*_{-1}(n, n') \mid 0\leq n' < n\}) = 0.$
    \item It is easy to confirm that $G^*_{-1}(0,2 )= G^*_{-1}(1,3) = G^*_{-1}(2, 0) = G^*_{-1}(3, 1) = 1.$ 
    Let $n \geq 2$.
    For any $n' < n, G^*_{-1}(2n', 2n + 1) \neq 1$ and $G^*_{-1}(2n' + 1, 2n + 1) \neq 1$ because $G^*_{-1}(2n', 2n' + 1) = 1$ if $n' \geq 2, G^*_{-1}(2n' + 1, 2n') = 1$ if $n' \geq 2,$ and $G^*_{-1}(0,2 )= G^*_{-1}(1,3) = G^*_{-1}(2, 0) = G^*_{-1}(3, 1) = 1.$ Similarly, for any $n' < n, G^*_{-1}(2n, 2n') \neq 1$ and $G^*_{-1}(2n, 2n' + 1)\neq 1$ hold. In addition, from (1), $G^*_{-1}(2n, 2n) = 0.$ Therefore, $G^*_{-1}(2n, 2n + 1) = {\rm mex}(\{G^*_{-1}(x', 2n + 1) \mid 0 \leq x' < 2n\} \cup \{G^*_{-1}(2n, y') \mid 0 \leq y' < 2n+ 1\}) = 1$. We can also show $G^*_{-1}(2n + 1, 2n) = 1$ by symmetry.
    \item It is easy to confirm that $G^*_{-1}(0, 3) = G^*_{-1}(1, 2) = G^*_{-1}(2, 1) = G^*_{-1}(3, 0) = 2.$ 
    Let $n \geq 1$.
    For any $n' < n, G^*_{-1}(4n', 4n + 2) \neq 2, G^*_{-1}(4n' +  1, 4n + 2) \neq 2, G^*_{-1}(4n' + 2, 4n + 2) \neq 2,$ and $ G^*_{-1}(4n' + 3, 4n + 2) \neq 2$ because $G^*_{-1}(4n', 4n' + 2) = 2$ if $n' \geq 1, G^*_{-1}(4n' + 1, 4n' + 3) = 2$ if $n' \geq 1, G^*_{-1}(4n' + 2, 4n') = 2$ if $n' \geq 1, G^*_{-1}(4n' + 3, 4n' + 1) = 2 $ if $n' \geq 1,$ and   $G^*_{-1}(0, 3) = G^*_{-1}(1, 2) = G^*_{-1}(2, 1) = G^*_{-1}(3, 0) = 2$. Similarly, for any $n' < n, G^*_{-1}(4n, 4n') \neq 2, G^*_{-1}(4n, 4n' + 1) \neq 2, G^*_{-1}(4n, 4n' + 2) \neq 2, $ and $G^*_{-1}(4n, 4n' + 3) \neq 2$ hold. Furthermore, from (1), $G^*_{-1}(4n, 4n) = 0$ and from (2), $G^*_{-1}(4n, 4n + 1) = 1$. Therefore, $G^*_{-1}(4n, 4n + 2) = {\rm mex}(\{G^*_{-1}(x', 4n + 2) \mid 0 \leq x' < 4n\} \cup \{G^*_{-1}(4n, y') \mid 0 \leq y' < 4n + 2\}) = 2$.
    We can also show $G^*_{-1}(4n + 1, 4n + 3) = G^*_{-1}(4n + 2, 4n) = G^*_{-1}(4n + 3, 4n + 1) = 2$ in similar ways.
    \item It is easy to confirm that $G^*_{-1}(0, 4) = G^*_{-1}(1, 5) = G^*_{-1}(2, 3) = G^*_{-1}(3, 2) = G^*_{-1}(4, 0) = G^*_{-1}(5, 1) = 3.$ 
    Let $n \geq 2$.
    For any $n' < n, G^*_{-1}(4n' -2, 4n) \neq 3, G^*_{-1}(4n' - 1, 4n) \neq 3, G^*_{-1}(4n', 4n) \neq 3,$ and $G^*_{-1}(4n' + 1, 4n) \neq 3$ because $G^*_{-1}(4n' - 2, 4n') = 3$ if $n' \geq 2, G^*_{-1}(4n' - 1, 4n' + 1) = 3$ if $n' \geq 2, G^*_{-1}(4n', 4n' -2) = 3$ if $n' \geq 2, G^*_{-1}(4n' + 1, 4n' - 1) = 3$ if $n' \geq 2,$ and $G^*_{-1}(0, 4) = G^*_{-1}(1, 5) = G^*_{-1}(2, 3) = G^*_{-1}(3, 2) = G^*_{-1}(4, 0) = G^*_{-1}(5, 1) = 3.$ Similarly, for any $n' < n, G^*_{-1}(4n - 2, 4n' - 2) \neq 3, G^*_{-1}(4n - 2, 4n' - 1) \neq 3, G^*_{-1}(4n - 2, 4n') \neq 3, $ and $G^*_{-1}(4n - 2, 4n' + 1) \neq 3$ hold.
    Furthermore, from (1), $G^*_{-1}(4n - 2, 4n - 2) = 0$, from (2), $G^*_{-1}(4n - 2, 4n - 1) = 1$, and from (3), $G^*_{-1}(4n-2, 4n-4) = 2.$ Therefore, $G^*_{-1}(4n-2, 4n) = {\rm mex}(\{G^*_{-1}(x', 4n) \mid 0 \leq x' < 4n - 2\} \cup \{G^*_{-1}(4n -2,y') \mid 0 \leq y' < 4n\}) = 3.$ We can also show $G^*_{-1}(4n - 1, 4n + 1) = G^*_{-1}(4n, 4n - 2) = G^*_{-1}(4n + 1, 4n - 1) = 3$ in similar ways.
\end{enumerate}
\end{proof}

We also show that the table of $G^*_{-1}$ can be separated into $2 \times 2$ blocks in which diagonally opposite numbers are equal to each other.

\begin{proposition}
For any nonnegative integers $n$ and $m$, $G^*_{-1}(2n, 2m) = G^*_{-1}(2n + 1, 2m + 1)$ and $G^*_{-1}(2n + 1, 2m) = G^*_{-1}(2n, 2m +1)$ hold.
\end{proposition}

\begin{proof}
    We prove this by induction.
    For the case of $n = m = 0$, $G^*_{-1}(0, 0) = G^*_{-1}(1, 1) = 0$  and $G^*_{-1}(0, 1) = G^*_{-1}(1, 0) = -1$. Thus, the statement holds.

    Consider the case $n + m > 0$.
    Then, from the definition of $G^*_{-1}$ and the induction hypothesis, $G^*_{-1}(2n, 2m) = {\rm mex}(\{G^*_{-1}(x', 2m) \mid 0\leq x' < 2n\} \cup \{G^*_{-1}(2n, y') \mid 0 \leq y' < 2m\}) = {\rm mex}(\{G^*_{-1}(2n', 2m), G^*_{-1}(2n'+ 1, 2m) \mid 0\leq n' < n\} \cup \{G^*_{-1}(2n, 2m'), G^*_{-1}(2n, 2m'+1) \mid 0\leq m' < m\}) = {\rm mex}(\{G^*_{-1}(2n'+1, 2m+1), G^*_{-1}(2n', 2m+1) \mid 0\leq n' < n\} \cup \{G^*_{-1}(2n+1, 2m'+1), G^*_{-1}(2n+1, 2m') \mid 0 \leq m' < m\})$.
    On the other hand, $G^*_{-1}(2n + 1, 2m + 1) = {\rm mex}(\{G^*_{-1}(2n'+1, 2m+1), G^*_{-1}(2n', 2m+1) \mid 0\leq n' < n\} \cup \{G^*_{-1}(2n+1, 2m'+1), G^*_{-1}(2n+1, 2m') \mid 0 \leq m' < m\} \cup \{G^*_{-1}(2n, 2m+1)\} \cup \{G^*_{-1}(2n + 1, 2m)\}).$ Since $G^*_{-1}(2n, 2m+1) \neq G^*_{-1}(2n, 2m)$ and $G^*_{-1}(2n+1, 2m) \neq G^*_{-1}(2n, 2m),$ we have
    $G^*_{-1}(2n + 1, 2m + 1) = G^*_{-1}(2n, 2m)$.

    In addition, from the definition of $G^*_{-1}$ and the induction hypothesis, $G^*_{-1}(2n+1, 2m) = {\rm mex}(\{G^*_{-1}(x', 2m)\mid 0 \leq x' < 2n+1\} \cup \{G^*_{-1}(2n+1, y') \mid 0 \leq y' < 2m\}) = {\rm mex}(\{G^*_{-1}(2n', 2m), G^*_{-1}(2n' +1, 2m) \mid 0 \leq n' < n\} \cup \{G^*_{-1}(2n, 2m)\} \cup \{G^*_{-1}(2n+1, 2m'), G^*_{-1}(2n + 1,2m'+1) \mid 0 \leq m' < m\}) = {\rm mex}(\{G^*_{-1}(2n'+1, 2m+1), G^*_{-1}(2n', 2m+1) \mid 0 \leq n' < n\} \cup \{G^*_{-1}(2n, 2m)\} \cup \{G^*_{-1}(2n, 2m'+1), G^*_{-1}(2n,2m') \mid 0 \leq m' < m\}) = G^*_{-1}(2n, 2m+1).$
\end{proof}

\section{Conclusion}
In this paper, we introduced two functions $G_{-1}$ and $G^*_{-1}$ from mis\`{e}re {\sc nim} and normal {\sc nim}. We showed that these functions can be used to determine which player has a winning strategy in two-player {\sc goishi hiroi} under normal and mis\`{e}re convention. For these functions, we showed necessary and sufficient conditions for the cases that the return values of the function are small.

In the future, if we will find more rulesets which have the same constructions, it becomes important to extend the theory such that analyzing games by using functions $G_{-1}$ and $G^*_{-1}$ as oracles.

\section*{Acknowledgments}
This work is supported by JSPS, Grant-in-Aid for Early-Career Scientists, Grant Number 22K13953.

\end{document}